%% file: DC-clean-main.tex
\documentclass[letterpaper,11pt]{article}

\usepackage{amsmath}
\usepackage{amsthm}
\usepackage{lmodern}
\usepackage[margin=1in]{geometry}
\usepackage{newtxmath}
\usepackage{prettyref}
\usepackage{xcolor}
\usepackage{xspace}
\usepackage[shortlabels]{enumitem}

\newcommand{\DC}{\textsc{DoubleCoverage}\xspace}

\newcommand{\cost}{\textit{cost}}
\newcommand{\costup}{\textit{cost}^{\uparrow}}
\newcommand{\Ind}{\vvmathbb 1}
\let\1\Ind

\newcommand{\TT}{\mathbb T}

\newcommand{\rt}{\vvmathbb{r}}
\newcommand{\dl}[1]{{\textcolor{orange}{#1}}}

\newtheorem{theorem}{Theorem}
\newtheorem{lemma}[theorem]{Lemma}
\newtheorem{obs}[theorem]{Observation}

\newtheorem{cor}[theorem]{Corollary}

\usepackage{fullpage}
\usepackage[textsize=small,textwidth=4cm,color=green]{todonotes}

\DeclareMathOperator{\lca}{lca}

\title{Online $k$-Taxi via Double Coverage and Time-Reverse Primal-Dual}
\author{
Niv Buchbinder\thanks{Dept. of Statistics and Operations Research, Tel Aviv University, Israel. Email: \texttt{niv.buchbinder@gmail.com}.}
\and
Christian Coester\thanks{CWI, Amsterdam, Netherlands. Email: \texttt{christian.coester@cwi.nl}.}
\and
Joseph (Seffi) Naor\thanks{Computer Science Department, Technion, Israel. Email: \texttt{naor@cs.technion.ac.il}.}
}

\date{}

\begin{document}

\begin{titlepage}
\maketitle
\thispagestyle{empty}

\begin{abstract}
	We consider the online $k$-taxi problem, a generalization of the $k$-server problem, in which $k$ servers are located in a metric space. A sequence of requests is revealed one by one, where each request is a pair of two points, representing the start and destination of a travel request by a passenger. The goal is to serve all requests while minimizing the distance traveled \emph{without carrying a passenger}.
	
	We show that the classic \emph{Double Coverage} algorithm has competitive ratio $2^k-1$ on HSTs, matching a recent lower bound for deterministic algorithms. For bounded depth HSTs, the competitive ratio turns out to be much better and we obtain tight bounds. When the depth is $d\ll k$, these bounds are approximately $k^d/d!$. By standard embedding results, we obtain a randomized algorithm for arbitrary $n$-point metrics with (polynomial) competitive ratio $O(k^c\Delta^{1/c}\log_{\Delta} n)$, where $\Delta$ is the aspect ratio and $c\ge 1$ is an arbitrary positive integer constant. The only previous known bound was $O(2^k\log n)$. For general (weighted) tree metrics, we prove the competitive ratio of Double Coverage to be $\Theta(k^d)$ for any fixed depth $d$, but unlike on HSTs it is not bounded by $2^k-1$.
	
	We obtain our results by a dual fitting analysis where the dual solution is constructed step-by-step \emph{backwards} in time. Unlike the forward-time approach typical of online primal-dual analyses, this allows us to combine information from the past and the future when assigning dual variables. We believe this method can be useful also for other problems. Using this technique, we also provide a dual fitting proof of the $k$-competitiveness of Double Coverage for the $k$-server problem on trees.
\end{abstract}
\end{titlepage}

\input{intro.tex}
\input{preliminaries.tex}

\input{lp-formulation.tex}
\input{dc-hst.tex}
\input{dc-weighted-trees.tex}
\input{lowerbound.tex}

\bibliographystyle{plain}
\bibliography{bibliography}

\appendix
\input{limitations.tex}

\end{document}

%% file: intro.tex
\section{Introduction}\label{sec:intro}

The $k$-taxi problem, proposed three decades ago as a natural generalization of the $k$-server problem by Fiat et al. \cite{FiatRR90}, has
gained renewed interest recently.
In this problem there are $k$ servers, or taxis, which are located in a metric space containing $n$ points. A sequence of requests is revealed one by one to an online algorithm, where each request is a pair of two points, representing the start and destination of a travel request by a passenger. An online algorithm must serve each request (by selecting a server that travels first to its start and then its destination) without knowledge of future requests. The goal is to minimize the total distance traveled by the servers \emph{without carrying a passenger}.
The motivation for not taking into account the distance the servers travel with a passenger is that any
algorithm needs to travel from the start to the destination, independently of the algorithm's decisions. Thus, the $k$-taxi
problem seeks to only minimize the overhead travel that depends on the algorithm's decisions. While this does not affect the optimal (offline) assignment, it affects the competitive factor.

Besides scheduling taxi rides, the $k$-taxi problem also models  tasks such as
scheduling elevators (the metric space is the line), and other applications where objects need to be transported
between locations.

The extensively studied and influential $k$-server problem is the special case of the $k$-taxi problem where for each request, the start equals the destination. A classical algorithm for the $k$-server problem on tree metrics is \DC.
This algorithm is described as follows. A server $s$ is called {\em unobstructed} if there is no other server on the unique path from $s$ to the current request. To serve the request, \DC moves \emph{all} unobstructed servers towards the request at equal speed, until one of them reaches the request. If a server becomes obstructed during this process, it stops while the others keep moving.

\DC was originally proposed for the line metric, to which it owes its name, as there are at most two servers moving at once. For a line metric it achieves the optimal competitive ratio of $k$ \cite{ChrobakKPV91}, and this result was later generalized to tree metrics~\cite{ChrobakL91}.

Given the simplicity and elegance of \DC, it is only natural to analyze its performance for the $k$-taxi problem. Here, we use it only for bringing a server to the start vertex of a request.

\subsection{Related Work and Known Results}
For the $k$-server problem, the best known deterministic competitive factor on general metrics is $2k-1$ \cite{KoutsoupiasP95}; with randomization, on hierarchically well-separated trees (HSTs)\footnote{See Section \ref{sec:pre} for an exact definition of HSTs.} the best known bound is $O(\log^2k)$ \cite{BCLLM18,BGMN19}. By a standard embedding argument, this implies a bound of $O(\log^2k\log n)$ for $n$-point metrics, and it was also shown in \cite{BCLLM18} that a dynamic embedding yields a bound of $O(\log^3 k\log \Delta)$ for metrics with aspect ratio $\Delta$. In \cite{Lee18}, a more involved dynamic embedding was proposed to achieve a polylog$(k)$-competitive algorithm for general metrics.\footnote{There is a gap in the version posted to the arXiv on February 21, 2018~\cite{Lee18arxiv,Lee19personal}.} Contrast these upper bounds with the known deterministic lower bound of $k$ \cite{ManasseMS88} and the randomized lower bound of $\Omega(\log k)$ \cite{FiatKLMSY91}. More information about the $k$-server problem can be found in \cite{Koutsoupias09}.

Surprisingly, until recently very little was known about the $k$-taxi problem, in contrast to the extensive work on the $k$-server problem.
Coester and Koutsoupias \cite{CK19} provided a $(2^k-1)$-competitive memoryless randomized algorithm for the $k$-taxi problem on HSTs against an adaptive online adversary. This result implies: (i) the existence of a $4^k$-competitive deterministic algorithm for HSTs via a known reduction~\cite{Ben-DavidBKTW94}, although this argument is non-constructive; (ii) an $O(2^k\log n)$-competitive randomized algorithm for general metric spaces (against an oblivious adversary). Both bounds currently constitute the state of the art.
Coester and Koutsoupias also provided a lower bound of $2^k-1$ on the competitive factor of any deterministic algorithm for the $k$-taxi problem on HSTs, thus proving that the problem is substantially harder than the $k$-server problem.
However, large gaps still remain in our understanding of the $k$-taxi problem, and many problems remain open in both deterministic and randomized settings. For general metrics, an algorithm with competitive factor depending only on $k$ is known only if $k=2$, and for the line metric only if $k\le 3$ \cite{CK19}. Both of these algorithms can be viewed as variants of \DC.

The version of the problem where the start-to-destination distances also contribute to the objective function was called the ``easy'' $k$-taxi problem in \cite{Kosoresow96,CK19}, whereas the version we are considering here is the ``hard'' $k$-taxi problem. The easy version has the same competitive factor as the $k$-server problem~\cite{CK19}. The $k$-taxi problem was recently reintroduced as the Uber problem in \cite{DehghaniEHLS17}, who studied the easy version in a stochastic setting.

\subsection{Our Contribution}
We provide the following bounds on the competitive ratio of \DC for the $k$-taxi problem.

\begin{theorem}\label{thm:hst}
The competitive ratio of \DC for the $k$-taxi problem is at most
\begin{enumerate}[(a)]
\item $\left(c_{kd}=\sum_{h=1}^{\min\{k,d\}}\binom{k}{h}\right)$ on HSTs of depth $d$.\label{it:ThmUbHst}
\item $O(k^d)$-competitive on general (weighted) tree metrics of depth $d$.\label{it:ThmUbTree}
\end{enumerate}
\end{theorem}

We complement these upper bounds by the following lower bounds:

\begin{theorem}\label{thm:lowerbound}
The competitive ratio of \DC for the $k$-taxi problem is at least
\begin{enumerate}[(a)]
\item $\left(c_{kd}=\sum_{h=1}^{\min\{k,d\}}\binom{k}{h}\right)$  on HSTs of depth $d$.\label{it:ThmLbHst}
\item $\Omega(k^d)$ on (even unweighted) tree metrics of constant depth $d$.\label{it:ThmLbTree}
\end{enumerate}
\end{theorem}
When the depth $d$ of the HST is at least $k$, the upper bound $c_{kd}= 2^k-1$ also matches exactly the lower bound from~\cite{CK19} that holds even for randomized algorithms against an adaptive online adversary. Note that for fixed $d$, $c_{kd}$ is roughly $k^d/d!$ up to a multiplicative error that tends to $1$ as $k\to\infty$. The $\Omega(k^d)$ lower bound on general trees is hiding a constant factor that depends on $d$. Since the root on general trees can be chosen arbitrarily, $d$ is essentially half the hop-diameter.

By well-known embedding techniques of general metrics into HSTs \cite{Bartal96,FakcharoenpholRT03}, slightly adapted to HSTs of bounded depth (see Theorem \ref{thm:embedding} in Section \ref{sec:pre}), we obtain the following result for general metrics.

\begin{cor}\label{cor:hst2}
	There is a randomized $O(k^c\Delta^{1/c}\log_{\Delta} n)$-competitive algorithm for the $k$-taxi problem for every $n$-point metric, where $\Delta$ is the aspect ratio of the metric, and $c\ge 1$ is an arbitrary positive integer. In particular, setting $c= \left\lceil\sqrt{\frac{\log \Delta}{\log k}}\right\rceil$, the competitiveness is $2^{O\big(\sqrt{\log k\log\Delta}\big)}\log_\Delta n$.
\end{cor}

Compared to the $O(2^k\log n)$ upper bound from \cite{CK19}, our bound has only a polynomial dependence on $k$ at the expense of some dependence on the aspect ratio. Since $c_{kd}\le 2^k-1$ for all $d$, we still also recover the same $O(2^k\log n)$ competitive factor. The bounds in Corollary~\ref{cor:hst2} actually hide another division by $(c-1)!$ if $c\le k$. Therefore, whenever $\Delta$ is at most $2^{O(k^2)}$ our bound yields an improvement.

\paragraph{Techniques.}
For the $k$-server problem, there exists a simple potential function analysis of \DC. The potential value depends on the relative distances of the server locations, which, in the $k$-taxi problem, can change arbitrarily by relocation requests even though the algorithm does not incur any cost. Therefore, such a potential cannot work for the $k$-taxi problem. In~\cite{CK19}, the $2^k-1$ upper bound for the randomized HST algorithm is proved via a potential function that is $2^k-1$ times the minimum matching between the online and offline servers. As is stated there, the same potential can be used to obtain the same bound for \DC when $k=2$, but it fails already when $k= 3$ (see Appendix~\ref{sec:purePotential}). Nonetheless, they conjectured that \DC achieves the competitive ratio of $2^k-1$ on HSTs.

We are able to prove that this is the case (and give the more refined bound of $c_{kd}$) with a primal-dual approach (which still uses an auxilary potential function as well). The primal solution is the output of the \DC algorithm. A dual solution is constructed to provide a lower bound on the optimal cost. The typical way a dual solution is constructed in the online primal-dual framework is forward in time, step-by-step, along with the decisions of the online algorithm (see e.g. \cite{BuchbinderN09,GuptaN14,AzarBCCCG0KNNP16}). By showing that the objective values of the constructed primal and dual solutions are within a factor $c$ of each other, one gets that the primal solution is $c$-competitive and the dual solution is $1/c$-competitive. For the LP formulation of the $k$-taxi problem we are considering, we show that a pure forward-time approach (producing a competitive dual solution as well) is doomed to fail:
\begin{theorem}\label{thm:forward}
	There exists no competitive online algorithm for the dual problem of the $k$-taxi LP as defined in Section~\ref{sec:LP}, even for $k=1$.
\end{theorem}
Our main conceptual contribution is a novel way to overcome this problem by constructing the dual solution backwards in time.
Our assignment of dual variables for time $t$ combines knowledge about the future \emph{and} the past: It incorporates knowledge about the future simply due to the time-reversal; knowledge about the past is also used because the dual assignments are guided by the movement of \DC, which is a forward-time (online) algorithm. Our method can be seen as a restricted form of dual fitting suitable for online algorithms that is more ``local" and hence easier to be analyzed step by step, similarly to primal-dual algorithms.
We believe that this time-reversed method of constructing a dual solution may be useful for analyzing additional online problems, especially when information about the future helps to construct better dual solutions. Using this technique, we also provide a primal-dual proof of the $k$-competitiveness of \DC for the $k$-server problem on trees. To the best of our knowledge, a primal-dual proof of this classical result has not been known before.
\begin{theorem}[\cite{ChrobakL91}]\label{thm:kserver}
	\DC is $k$-competitive for the $k$-server problem on trees.
\end{theorem}

\subsection{Organization}
Section~\ref{sec:pre} defines necessary notation and terminology. In Section~\ref{sec:LP}, we provide a linear programming formulation of the $k$-taxi problem on trees and the corresponding dual, which we then transform to a more intuitive equivalent dual LP. In Section~\ref{sec:hst} we prove the upper bound for $k$-taxi on HSTs (Theorem~\ref{thm:hst}\ref{it:ThmUbHst}) and, as a byproduct, for $k$-server on general trees (Theorem~\ref{thm:kserver}). The upper bound for $k$-taxi on general trees  (Theorem~\ref{thm:hst}\ref{it:ThmUbTree}) is given in Section~\ref{sec:generaltree}. The lower bounds (Theorem~\ref{thm:lowerbound}) are proved in Section~\ref{sec:Lb}. Corollary~\ref{cor:hst2} follows directly from Theorem~\ref{thm:hst}\ref{it:ThmUbHst} and Theorem~\ref{thm:embedding} in Section~\ref{sec:pre}.
Appendix \ref{sec:noForward} contains a proof of the limitation of forward-time setting of dual variables (Theorem~\ref{thm:forward}).

%% file: preliminaries.tex
\section{Preliminaries}\label{sec:pre}
\paragraph{The $k$-Taxi Problem.} The $k$-taxi problem is formally defined as follows. We are given a metric space with point set $V$, where $|V|=n$. Initially, $k$ taxis, or servers, are located at points of $V$. At each time $t$ we get a request $(s_t,d_t)$, where $s_t, d_t \in V$. The request is served by moving one of the servers to $s_t$ (unless there is already a server at $s_t$). The cost paid by the algorithm is the distance traveled by the server to $s_t$. Then, one of the servers from $s_t$ is relocated to the point $d_t$. There is no cost for relocating the server from $s_t$ to $d_t$. The goal is to minimize the cost.

Without loss of generality, we can split each request $(s_t,d_t)$ into two requests: a {\em simple} request and a {\em relocation} request. Thus, at each time $t$, request $(s_t,d_t)$ is either one of the following:
\begin{itemize}
\item {\em Simple request ($s_t=d_t$)}: a server needs to move to $s_t$, if there is no server there already. The cost is the distance traveled by the server.
\item {\em Relocation request ($s_t=d_{t-1}$)}: a server is relocated from $s_t$ to $d_t$. There is no relocation cost.
\end{itemize}
We can then partition the time horizon into two sets $T_s$, $T_r$ (odd and even times). For times in $T_s=\{1,3,5, \ldots, 2T-1\}$ we have simple requests, and for times in $T_r=\{2,4,6, \ldots, 2T\}$ we have relocation requests. The $k$-server problem is the special case of  the $k$-taxi problem without relocation requests.

\paragraph{Trees and HSTs.} Consider now a tree $\TT=(V,E)$ and let $\rt$ denote its root.
There is a positive weight function defined over the edge set $E$, and without loss of generality  all edge weights are integral.
The {\em distance} between vertices $u$ and $v$ is the sum
 of the weights of the edges on the (unique) path between them in ${\TT}$, which induces a metric.
 The {\em combinatorial depth} of a vertex $v \in V$ is defined to be the number of edges on the path from $\rt$ to $v$.
 The combinatorial depth of ${\TT}$ is the maximum combinatorial depth among all vertices.
At times it will be convenient to assume  that all edges in $E$ have unit length by breaking edges into unit length parts called {\em short edges}. We then refer to the original edges of $\TT$ as {\em long edges}. However, the combinatorial depth of ${\TT}$ is still defined in terms of the long edges. We define the {\em weighted depth} of a vertex $u$ as the number of \emph{short} edges on the path from $\rt$ to $u$. For $u \in V$, let $V_u\subseteq V$ be the vertices of the subtree rooted at $u$. In trees where all the leaves are at the same weighted/combinatorial depth (namely, HSTs, see below), we define the {\em weighted/combinatorial height} of a vertex $u$ as the number of short/long edges on the path from $u$ to any leaf in $V_u$. We denote by $v\prec u$ that $v$ is a child of $u$. We write $p(u)$ for the parent of $u$.

{\em Hierarchically well-separated trees (HSTs)}, introduced by Bartal
\cite{Bartal96}, are special trees that can be used to approximate arbitrary finite metrics. For $\alpha\ge 1$, an $\alpha$-HST is a tree where every leaf is at the same combinatorial depth $d$ and the edge weights along any root-to-leaf path decrease by a factor $\alpha$ in each step. The associated metric space of the HST is only the set of its leaves. Hence, for the $k$-taxi problem on HSTs, the requested points $s_t$ and $d_t$ are always leaves. Any $n$-point metric space can be embedded into a random $\alpha$-HST such that (i) the distance between any two points can only be larger in the HST and (ii) the expected blow-up of each distance is $O(\alpha\log_\alpha n)$ \cite{FakcharoenpholRT03}. The latter quantity is also called the \emph{distortion} of the embedding. The depth of the random HST constructed in the embedding is at most $\lceil\log_\alpha \Delta\rceil$, where $\Delta$ is the aspect ratio, i.e., the ratio between the longest and shortest non-zero distance. Choosing $\alpha=\Delta^{1/d}$, we obtain an HST of depth $d$ and with distortion $O\left(d\Delta^{1/d}\log_\Delta n\right)$.

\begin{theorem}[Corollary to \cite{FakcharoenpholRT03}]\label{thm:embedding}
Any metric with $n$ points and aspect ratio $\Delta$ can be embedded into a random HST of combinatorial depth $d$ with distortion $O\left(d\Delta^{1/d}\log_\Delta n\right)$.
\end{theorem}

\paragraph{The Double Coverage Algorithm.}
We define \DC in a way that will suit our definition of short edges of length $1$ later. Consider the arrival of a simple request at location $s_t$. A server located at vertex $v$ is {\em unobstructed} if there are no other servers on the path between $v$ and $s_t$. If other servers on this path exist only at $v$, we consider only one of them unobstructed (chosen arbitrarily). Serving the request is done in several small steps, as follows:\\

\vspace{0.15cm}
\noindent \framebox[1.05\width][l]{
\begin{minipage}{0.94\linewidth}
\DC (upon a simple request at $s_t$):\\
While no server is at $s_t$: all currently unobstructed servers move distance $1$ towards $s_t$.
\end{minipage}}
\vspace{0.15cm}
\\

\noindent Upon a relocation request $(s_t,d_t)$, we simply relocate a server from $s_t$ to $d_t$.

For a given small step (i.e., iteration of the while-loop), we denote by $U$ and $B$ the sets of servers moving towards the root (\textbf{u}pwards in the tree) and away from the root (towards the \textbf{b}ottom of the tree), respectively.
\begin{obs}
	In any small step:
	\begin{itemize}
		\item $B$ is either a singleton ($B=\{j\}$ for a server $j$) or empty ($B=\emptyset$).
		\item The subtrees rooted at servers of $U$ are disjoint and do not contain $s_t$. If $B=\{j\}$, then these subtrees and $s_t$ are inside the subtree rooted at $j$.
	\end{itemize}
\end{obs}

%% file: lp-formulation.tex
\section{LP formulations}\label{sec:LP}

We formulate a linear program (LP) for the $k$-taxi problem on trees along with its dual that we use for the purpose of analysis. We assume for ease of exposition that all edges are short edges. The formulation is a relaxation of the problem as it allows for fractional values of the variables\footnote{In general, for any metric, the offline $k$-taxi problem can be solved in polynomial time by a reduction to a min-cost flow problem.}.
For $u \in V$, let variable $x_{ut}$ denote the number of servers in $V_u$ after the request at time $t$ has been served. Variable $y_{ut}\geq 0$ denotes the number of servers that left subtree $V_u$ (moving upwards) at time $t$. The variable $z_{ut}\ge 0$ denotes the number of servers that enter $V_u$ (moving downwards) at time $t$.
For $u=\rt$, $x_{\rt t}$ is defined to be the \emph{constant} $k$. (It is not a variable.) The dual variables corresponding to the primal constraints appear in parenthesis  left of the constraints.
The primal LP is the following:
\begin{alignat*}{3}
&&\text{min }~~~   & \sum_{t\in T_s} \sum_{u\ne \rt} (y_{ut}+z_{ut})\\
&\dl{(\lambda_{ut})}\qquad& & x_{ut}\ge \1_{\{u=s_t\text{ and }t\in T_s\}} + \sum_{v\prec u} x_{vt} \quad&& \forall u\in V, t\in T_s \cup T_r \\
&\dl{(b_{u,t-1})}&& y_{ut}-z_{ut} = x_{u,t-1}-x_{ut} \quad&&  \forall u\ne \rt, t \in T_s\\
&\dl{(b_{u,t-1})}&&x_{ut} = x_{u,t-1}+ \xi_{ut} \quad&&  \forall u\ne \rt, t \in T_r\\
&&& y_{ut}, z_{ut}\ge 0 && \forall u\ne \rt, t \in T_s,
\end{alignat*}
where
\[\xi_{ut} := \left\{\begin{array}{ll}
-1 & \text{if } s_t \in V_u, d_t \notin V_u\\
1 & \text{if } s_t \notin V_u, d_t \in V_u\\
0 & \mbox{otherwise.}
\end{array} \right.\]
For technical reasons, before we construct the dual LP we will add the additional constraint
\begin{align*}
\dl{(b_{u,2T})}\qquad\qquad 0= x_{u,2T} - \bar{x}_{u,2T} \qquad \forall u\ne\rt
\end{align*}
to the primal LP, where $\bar{x}_{u,2T}$ are \emph{constants} specifying the configuration of \DC at the last time step. Clearly, this affects the optimal value by only an additive constant. We will also view $x_{u0}=\bar{x}_{u0}$ as constants describing the initial configuration of the servers. The corresponding dual LP is the following.
\begin{align}
 \lefteqn{\hspace*{-1cm}\max~~~ \sum_{t\in T_s} \left(\lambda_{s_tt}-k\lambda_{\rt t}\right) + \sum_{t\in T_r} \left(-k\lambda_{\rt t}+ \sum_{u\ne \rt}\xi_{ut}b_{u,t-1}\right) ~+~ \sum_{u\ne \rt} \bar{x}_{u0} b_{u0} - \sum_{u\ne \rt} \bar{x}_{u,2T} b_{u,2T}} \nonumber \\ \nonumber\\
&\dl{(x_{ut})} && \lambda_{ut}-\lambda_{p(u)t} = b_{ut}-b_{u,t-1} & \forall u\ne \rt, t \in T_s \cup T_r \label{dual1Both} \\
&\dl{(z_{ut}, y_{ut})}&& b_{u,t-1} \in[-1, 1] &  \forall u\ne \rt, t \in T_s \label{dual2Both} \\
&&& \lambda_{ut}\ge 0\quad&  \forall u\in V, t\in T_s\cup T_r \label{dual3Both}
\end{align}
We can use the same primal and dual formulation for the $k$-server problem, except that the set $T_r$ is then empty.\footnote{We note that our LP for the $k$-server problem is different from LPs used in the context of polylogarithmically-competitive randomized algorithms for the $k$-server problem. In our context of deterministic algorithms for $k$-taxi (and $k$-server), we show that we can work with this simpler formulation.}

As already mentioned, when considering the $k$-taxi problem on HSTs, requests appear only at the leaves. It is easy to see that in this case the upward movement cost (i.e., movement towards the root) is the same as the downward movement cost, up to an additive error of $k$ times the distance from the root to any leaf. The same is true of the $k$-server problem on general trees (but not for the $k$-taxi problem on general trees).
Hence, for the $k$-taxi problem on HSTs and for the $k$-server problem, we can use an LP that only takes into account the upward movement cost, and thus the coefficient of the variables $z_{ut}$ in the primal objective function become zero. The only change to the dual linear program as a result of this is that constraint \eqref{dual2Both} becomes:
\begin{equation}b_{u,t-1} \in[0, 1] \qquad  \forall u\ne \rt, t \in T_s \qquad \mbox{when measuring only upward cost}\label{dual2} \end{equation}

\subsection{Dual transformation}

The dual LP is not very intuitive. By a transformation of variables, we get a simpler equivalent dual LP, which we can interpret as building a mountain structure on the tree over time. This new dual LP has only one variable $A_{ut}$ for each vertex $u$ and time $t$. We interptet  $A_{ut}$ as the \emph{altitude} of $u$ at time $t$. We denote by $\Delta_tA_{u}= A_{ut}- A_{u,t-1}$ the change of altitude of vertex $u$ at time $t$. For a server $i$ of \DC, we denote by $v_{it}$ its location at time $t$, and define similarly $\Delta_tA_i:=A_{v_{it}t}-A_{v_{i,t-1}t-1}$ as the change of altitude of server $i$ at time $t$. The new dual LP is the following:

\begin{align}
\lefteqn{\hspace*{-1cm}\max~~~ \sum_{t\in T_s}\left[\Delta_tA_{s_t}-\sum_{i=1}^{k} \Delta_t A_{i}\right] - \sum_{t\in T_r} \sum_{i=1}^{k} \Delta_tA_{v_{it}}} \nonumber \\ \nonumber\\
	& A_{ut}-A_{p(u)t} \in[-1, 1] &  \forall u\ne \rt, t+1 \in T_s \label{eq:slope}\\
	& \Delta_tA_{u}\ge 0\quad&  \forall u\in V, t\in T_s\cup T_r\label{eq:AltGrow}
\end{align}

The constraints of the LP stipulate that the altitude of each node can only increase over time, and the absolute difference in altitude of two adjacent nodes is at most $1$ (at time steps before a simple request).
The objective function measures changes in the altitudes of request and server locations.
When measuring only movement towards the root, constraint \eqref{eq:slope} becomes:
\begin{align} A_{ut}-A_{p(u)t} \in[0, 1] &  \qquad \forall u\ne \rt, t+1 \in T_s \label{eq:slopeMon}\end{align}
This corresponds to the additional requirement that altitudes are non-decreasing along root-to-leaf paths.

We define
\begin{align}
D_t:= &  \begin{cases}\Delta_tA_{s_t}-\sum_{i=1}^{k} \Delta_t A_{i}\qquad& t\in T_s\\
-\sum_{i=1}^{k} \Delta_tA_{v_{it}} & t\in T_r,
\end{cases}\label{eq:Dt}
\end{align}
so that the dual objective function is equal to $D:=\sum_{t \in T_s \cup T_r} D_t$.

The following lemma allows us to use this new LP for our analyses.
\begin{lemma}
The two dual LPs are equivalent. That is, any feasible solution to one of them can be translated (online) to a feasible solution to the other with the same objective function value.
\end{lemma}
\begin{proof}
We refer to the first LP as the original LP and the second LP as the new LP.
Given a solution $A_{ut}$ to the new LP let
\begin{align*}
\lambda_{ut}&:=\Delta_tA_u:=A_{ut}-A_{u,t-1}\\
b_{ut}&:= A_{ut}-A_{p(u)t}.
\end{align*}

It is easy to see that feasibility for the new LP implies feasibility for the original LP; in particular, $\lambda_{ut}-\lambda_{p(u)t}=A_{ut}-A_{u,t-1}-A_{p(u)t}+A_{p(u),t-1}=b_{ut}-b_{u,t-1}$.
For the other direction, a feasible solution $\lambda_{ut}, b_{ut}$ to the original LP can be transformed by setting
\begin{align*}
A_{u t}:= \sum_{\tau=1}^t\lambda_{u \tau} + \sum_{v\ne \rt|u\in V_v}b_{v0}.
\end{align*}
Using this definition, we have:
\begin{align}
A_{ut}-A_{p(u)t} & =\sum_{\tau=1}^t(\lambda_{u \tau}-\lambda_{p(u)\tau}) + b_{u0} = b_{ut}\label{eq:Ab}\\
\Delta_t A_{u}& = A_{u t}- A_{u,t-1} = \lambda_{ut}\label{eq:lambda}
\end{align}
Again, feasibility for the new LP follows from the constraints of the original LP.

Finally, we consider the value of the objective function. Notice that equations \eqref{eq:Ab} and \eqref{eq:lambda} hold for both directions of the transformation, so we can use them in the following. Let $\bar{x}_{ut}$ denote the number of servers $i$ with $v_{it}\in V_u$. Using this, we can write
\begin{align}
\sum_{u\ne \rt}\bar{x}_{ut} b_{ut}&=\sum_{u\ne \rt}\sum_{i|v_{it}\in V_u} b_{ut}\nonumber\\
&= \sum_{i} \sum_{u\ne \rt|v_{it}\in V_u}b_{ut}\nonumber\\
&= \sum_i (A_{v_{it}t}-A_{\rt t})\nonumber\\
&=  - kA_{\rt t}+\sum_i A_{v_{it}t},\nonumber
\end{align}
where the penultimate equation follows from equation \eqref{eq:Ab}. Therefore,
\begin{align}
\sum_{u\ne \rt} \left(\bar{x}_{u,t-1} b_{u,t-1}-\bar{x}_{ut}b_{ut}\right) = k\Delta_t A_{\rt} -\sum_i \Delta_tA_i = k \lambda_{\rt t}-\sum_i \Delta_tA_i.\label{eq:xutbut}
\end{align}

Thus, for $t\in T_s$ we get
\begin{align}
D_t=\Delta_tA_{s_t}-\sum_i \Delta_t A_{i} = \lambda_{s_t t} -k \lambda_{\rt t} + \sum_{u\ne \rt} \left(\bar{x}_{u,t-1} b_{u,t-1}-\bar{x}_{ut}b_{ut}\right)\label{eq:DtTs}
\end{align}

For a relocation request at time $t\in T_r$, most servers stay in their old position (hence $\Delta_tA_i=\Delta_tA_{v_{it}}$ for these servers), except for the relocated server $i^*$, which moves from $v_{i^*,t-1}=s_t$ to $v_{i^*t}=d_t$. We therefore have
\begin{align*}
D_t= -\sum_i \Delta_tA_{v_{it}} =A_{d_t,t-1} - A_{s_t,t-1} - \sum_i \Delta_tA_i\qquad\qquad (\text{if }t\in T_r).
\end{align*}
Using \eqref{eq:Ab}, we can rewrite
\begin{align*}
A_{d_t,t-1} - A_{s_t,t-1}  &= \sum_{u | s_t \notin V_u, d_t \in V_u}b_{u,t-1} - \sum_{u | s_t \in V_u, d_t \notin V_u}b_{u,t-1} = \sum_{u\ne\rt}\xi_{ut} b_{u,t-1}
\end{align*}
Recalling \eqref{eq:xutbut}, we get for $t\in T_r$ that
\begin{align}
D_t= \sum_{u\ne\rt}\xi_{ut} b_{u,t-1} -k \lambda_{\rt t} + \sum_{u\ne \rt} \left(\bar{x}_{u,t-1} b_{u,t-1}-\bar{x}_{ut}b_{ut}\right).\label{eq:DtTr}
\end{align}

Combining \eqref{eq:DtTs} and \eqref{eq:DtTr} and noticing that the terms $\bar{x}_{u,t-1} b_{u,t-1}-\bar{x}_{ut}b_{ut}$  telescope over time for each $u\ne \rt$, the dual objective value of the new LP can be rewritten as
\begin{align*}
\sum_{t\in T_s\cup T_r} D_t& =  \sum_{t\in T_s} \left(\lambda_{s_tt}-k\lambda_{\rt t}\right) + \sum_{t\in T_r} \left(-k\lambda_{\rt t}+ \sum_{u\ne \rt}\xi_{ut}b_{u,t-1}\right) + \sum_{u\ne \rt} \left(\bar{x}_{u0} b_{u0}-\bar{x}_{u,2T} b_{u,2T}\right),
\end{align*}
which is precisely the objective value of the original dual.
\end{proof}

%% file: dc-hst.tex
\section{\texorpdfstring{The $k$-taxi Problem on HSTs}{The k-taxi Problem on HSTs}}\label{sec:hst}

In this section we analyze \DC on HSTs, proving part part \ref{it:ThmUbHst} of Theorem \ref{thm:hst}. As a byproduct we also give a primal-dual proof of the $k$-competitiveness of \DC for the $k$-server problem on trees (Theorem \ref{thm:kserver}).

Besides constructing a dual solution, our analysis will also employ a potential function $\Psi$. The choice of $\Psi$ will be the only difference between the analyses for the $k$-server and $k$-taxi problem.
The dual solution and potential will be such that for all $t\in T_s \cup T_r$,
\begin{align}
\costup_t +\Psi_t-\Psi_{t-1} \le c \cdot D_t,\label{eq:pot}
\end{align}
where $\costup_t$ is the cost of movement \emph{towards the root} by \DC's servers while serving the $t$th request, $c$ is the desired competitive ratio, and $D_t$ is the increase of the dual objective function at time $t$, as given by \eqref{eq:Dt}. As discussed in Section \ref{sec:pre} we may use in this case the dual of the program that only measures movement cost towards the root. Thus, summing \eqref{eq:pot} for all times will then imply that \DC is $c$-competitive.

Recall that for a simple request ($t\in T_s$), \DC breaks the movement of the servers
into small steps in which the servers in $U\cup B$ move distance 1 towards the request. We will break the construction of a dual solution into these same small steps. We will denote by $\Delta \Psi$ the change of $\Psi$ during the step and by $\Delta D$ the contribution of the step to $D_t$. The cost paid by the servers (for moving towards the root) in the step is $|U|$. Using this notation, we satisfy \eqref{eq:pot} for simple requests if we show for each step that:
\begin{align}
|U| + \Delta \Psi  &\le c \cdot \Delta D.\label{eq:potStep}
\end{align}

In Section \ref{sec:pd} we describe how we construct the dual solution by going backwards in time.
In Section \ref{sec:potkserver} we design a simple potential function that proves the $k$-competitiveness of $k$-server on trees.
In Section \ref{sec:potktaxi} we describe a more involved potential function proving the competitiveness for $k$-taxi on HSTs.

\subsection{Constructing the Dual Solution}\label{sec:pd}

As already mentioned, we break the construction of a dual solution into the same small steps that already partition the movement of $\DC$. That is, we will define altitudes also for the times between two successive small steps. We will call a dual solution where altitudes are also defined for times between small steps an \emph{extended dual solution}.

We will construct this dual solution by induction backwards in time. For a given point in time, let $A_u$ be the altitude of a vertex $u$ at this time as determined by the induction hypothesis, and let $v_i$ be the location of server $i$ at this time. We will denote by $A_u'$ and $v_i'$ the new values of these quantities at the next point in reverse-time. We denote by $\Delta A_u:=A_u-A_u'$ and $\Delta A_i:=A_{v_i}-A_{v_i'}'$ the change of the altitude of vertex $u$ and server $i$, respectively, in \emph{forward-time} direction. For the update due to a small step when serving a simple request $s_t$, define
\begin{align*}
\Delta D:= \Delta A_{s_t}-\sum_i \Delta A_{i}.
\end{align*}
Thus, the sum of the quantities $\Delta D$ for all small steps corresponding to a simple request at time $t$ is precisely $D_t$. In reverse-time, we can think of $\Delta D$ as the amount by which the request's altitude \emph{decreases} \emph{plus} the amount by which the sum of server altitudes \emph{increases}.

We will update altitudes so as to satisfy the following two rules:
\begin{enumerate}[(i)]
	\item $\Delta A_{u} \ge 0$ for all $u\in V$ (constraint \eqref{eq:AltGrow} is satisfied): In reverse-time, we only \emph{decrease} the altitude of any vertex (or leave it unchanged).\label{it:r1}
	\item $A_u-A_{p(u)}\in\{0,1\}$ for all $u\ne \rt$ at all times (constraint \eqref{eq:slopeMon} is satisfied):  The altitude of $u$ and $p(u)$ is the same, or the altitude of $u$ is higher by one than the altitude of $p(u)$. Overall, altitudes are non-increasing towards the root.\label{it:r2}
\end{enumerate}

\begin{lemma}\label{lem:backw}
	There exists a feasible extended dual solution satisfying:
	\begin{itemize}
		\item For a relocation request at time $t$: $D_t=0$.
		\item For a small step where $B=\emptyset$: $\Delta D\ge 1$.
		\item For a small step where $B=\{j\}$: $\Delta D\ge 0$.
	\end{itemize}
\end{lemma}
\begin{proof}
For the base of the reverse time induction, let $A$ be some arbitrary constant and define the altitude of every vertex $u\in V$ at the time after the final request to be $A$. This trivially satisfies rule \ref{it:r2}.

\paragraph{Relocation requests ($t\in T_r$):} We guarantee $D_t=0$ by simply keeping all altitudes unchanged.

\paragraph{Simple requests ($t\in T_s$):} Consider a small step of the simple request to $s_t$. In reverse-time, any server $i\in U$ moves from $v_i=p(v_i')$ to $v_i'$ during the small step. Then, $\Delta A_{i}= A_{p(v_i')}-A_{v_i'}'$. By rule \ref{it:r2} of the induction hypothesis, we have $A_{p(v_i')}-A_{v_i'}\in\{-1,0\}$. Similarly, if $B=\{j\}$, then $j$ moves from $v_j$ to $v_j'=p(v_j)$ in reverse-time, and $\Delta A_{j}= A_{v_j}-A_{p(v_j)}'$.

\paragraph{Case 1: $B=\emptyset$:}

	\begin{figure}
		\begin{center}
			\includegraphics{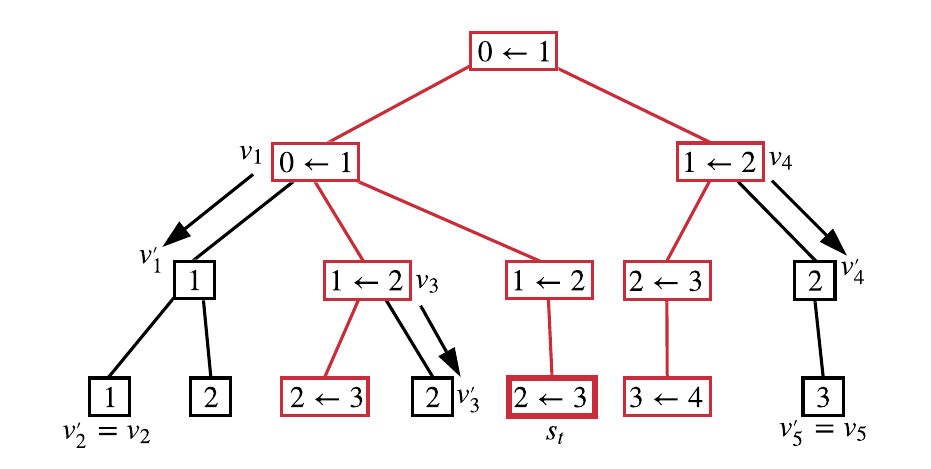}
			\caption{Example of a dual update for a small step of a simple request when $B=\emptyset$. The current request is $s_t$, and the vertices colored red are the component $V'$. The arrows show the movement of servers (in reverse time direction). Numbers in boxes represent altitudes, where $b \gets a$ means that the altitude is updated from $a$ to $b$ (in reverse time).}
			\label{fig:dual-upade}
		\end{center}
	\end{figure}

If for at least one server $i\in U$ we have $A_{v_i'}-A_{p(v_i')}=1$, then we set $A_{u}':= A_{u}$ for all vertices. In this case, $\Delta_t A_{i} =-1$ for the aforementioned server $i$, and for all other servers in $i\in U$, $\Delta_t A_{i} \leq 0$. Overall, $\Delta D\geq 1$.\\
Otherwise, for all $i\in U$, $A_{v_i'}-A_{p(v_i')}=0$ meaning that every edge along which a server moves during this small step connects two vertices of the same altitude (an example of the update of the dual for this case is shown in Figure~\ref{fig:dual-upade}). 
Let $V'=V\setminus \bigcup_{i\in U} V_{v_i'}$ be the connected component containing $s_t$ when cutting all edges traversed by a server in this step. Notice that $V'$ does not contain $v_i'$ even for servers $i$ that are not moving during the step, since those are located in subtrees below the servers of $U$. For each $u\in V'$, we set $A_{u}':= A_{u}-1$ (or $\Delta A_{u}=1$), and otherwise we keep the altitudes unchanged. In particular, rule \ref{it:r1} is satisfied. Since for all cut edges $A_{v_i'}-A_{p(v_i')}=0$, then for these edges $A_{v_i'}'-A_{p(v_i')}'=1$, and rule \ref{it:r2} also remains satisfied. As stated, servers only moved along edges connecting vertices of the same altitude (before the update), and all server positions $v_i'$ are outside the component $V'$, so $\Delta A_i=0$ for each server. But the component contains the request $s_t$, so $\Delta A_{s_t}= A_{s_t}- A_{s_t}'= 1$. Overall, we get $\Delta D=1$. 

\paragraph{Case 2: $B=\{j\}$:} 
	\begin{figure}
		\begin{center}
			\includegraphics{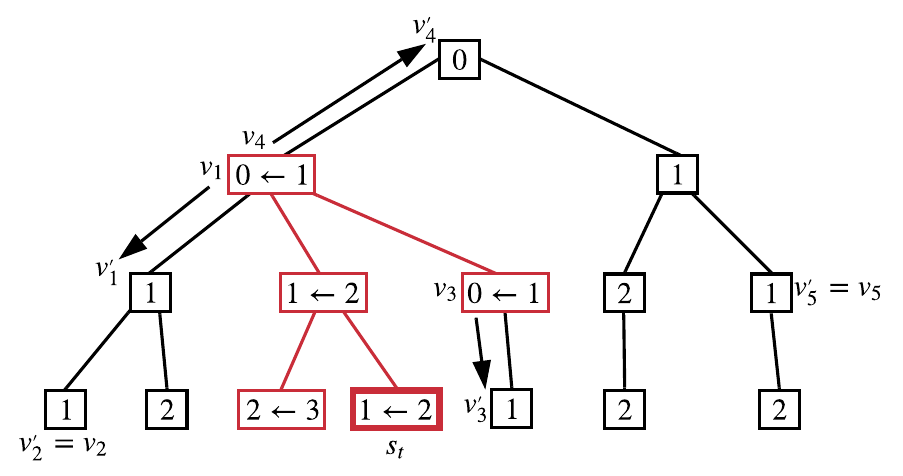}
			\caption{Example of a dual update for a small step of a simple request when $B=\{j\}$. The current request is $s_t$, and the vertices colored red are the component $V'$. The arrows show the movement of servers (in reverse time direction). Numbers in boxes represent altitudes, where $b \gets a$ means that the altitude is updated from $a$ to $b$ (in reverse time).}
			\label{fig:dual-upade2}
		\end{center}
	\end{figure}
If some server $i\in U$ moves in reverse-time to a vertex of higher altitude ($A_{p(v_i')}-A_{v_i'}=-1$) \emph{or} server $j$ moves to a vertex of the same altitude ($A_{p(v_j)}-A_{v_j}=0$), then we set $A_{u}':= A_{u}$ for each $u\in V$. In this case, $\Delta A_{i}\in\{-1,0\}$ for all $i\in U$, and $\Delta A_{j}\in\{0,1\}$, and the aforementioned condition translates to the condition that $\Delta A_{i}=-1$ for some $i\in U$ \emph{or} $\Delta A_{j}=0$.
Either case then guarantees that $\Delta D\ge 0$.\\
Otherwise, $j$ moves (in reverse-time) to a vertex of lower altitude ($A_{p(v_j)}-A_{v_j}=-1$) \emph{and} all servers in $U$ move along edges of unchanging altitude ($A_{p(v_i')}-A_{v_i'}=0$) (an example of the update of the dual for this case is shown in Figure~\ref{fig:dual-upade2}). Let $V'=V_{v_j}\setminus \bigcup_{i\in U} V_{v_i'}$ be the connected component containing $s_t$ when cutting all edges traversed by servers in this step. We decrease the altitudes of all vertices in this component by $1$ ($A_{u}':= A_{u}-1$ for $u\in V'$) and leave other altitudes unchanged, satisfying rule~\ref{it:r1}. As $A_{p(v_j)}-A_{v_j}=-1$ and $A_{p(v_i')}-A_{v_i'}=0$ for $i\in U$, also rule \ref{it:r2} is satisfied. Again, the locations $v_i'$ of any server $i$ (moving or not) are outside the component, so the update of altitudes does not affect $\Delta A_i$. Thus, $\Delta A_j=A_{v_j}-A_{p(v_j)}=1$ and, for each server $i\ne j$, $\Delta A_{i}= 0$. But the altitude of $s_t$ is decreasing by $1$ in reverse-time, so $\Delta_tA_{s_t}= 1$. Overall, we get that $\Delta D=0$.
\end{proof}

\paragraph{Potential Function Requirements.} Based on Lemma~\ref{lem:backw}, we conclude that the following requirements of a potential function $\Psi$ are sufficient to conclude inequality \eqref{eq:pot} (resp. \eqref{eq:potStep}) and therefore $c$-competitiveness of \DC.

\begin{obs}\label{obs:phireq}
\DC is $c$-competitive if there is a potential $\Psi$ satisfying:
\begin{itemize}
\item For a relocation request at time $t$: $\Psi_t=\Psi_{t-1}$.
\item In a single step of a simple request, where no server is going downwards: $|U| + \Delta \Psi \leq c$.
\item In a single step of a simple request, where there is a server moving downwards: $|U| + \Delta \Psi \leq 0$.
\end{itemize}
\end{obs}

\subsection{\texorpdfstring{Finalizing the Analysis for $k$-Server on Trees}{Finalizing the Analysis for k-Server on Trees}}\label{sec:potkserver}

In the $k$-server problem there are no relocation requests, and hence we only need to satisfy the two requirements involving simple requests in Observation \ref{obs:phireq}. We define the following potential function:
\begin{align*}
\Psi=-\sum_{i<j}d_{\lca(i,j)},
\end{align*}
where the sum is taken over all pairs of servers $\{i,j\}$ and $d_{\lca(i,j)}$ denotes the weighted depth (distance from the root $\rt$) of the least common ancestor of $i$ and $j$.

\paragraph{Case 1: $B=\emptyset$:}

In this case, the depth of the least common ancestor decreases by $1$ only for $i\in U$ and $j$ which is located in the subtree of $i$.  The total number of servers that are located in the subtrees below servers of $U$ is precisely $k-|U|$, and hence the potential function in this case grows by $k-|U|$. Thus, $|U| + \Delta \Psi = k$, satisfying Observation \ref{obs:phireq}.

\paragraph{Case 2: $B=\{j\}$:}

In this case, all the servers of $U$ are located in the subtree of $j$. For each server $i \in U$, the depth of the least common ancestor with $j$ increases by $1$, therefore contributing $-|U|$ to $\Delta \Psi $. There may be more servers in the subtree of $j$ that do not belong  to $U$. For each such server $i'$, there is a server $i \in U$ on its path to $j$. The depth of the least common ancestor of $i'$ and $i$ decreases by $1$, while the depth of the least common ancestor of $i'$ and $j$ increases by $1$. Thus, the total contribution of $i'$ to $\Delta \Psi $ is $0$, and likewise for servers outside the subtree of $j$. Hence,
$|U| + \Delta \Psi = 0$, satisfying Observation \ref{obs:phireq}.

\subsection{\texorpdfstring{Finalizing the Analysis for $k$-Taxi on HSTs}{Finalizing the Analysis for k-Taxi on HSTs}}\label{sec:potktaxi}

We show that the competitive ratio of the $k$-taxi problem on HSTs of combinatorial depth $d$ is
\begin{align*}
c_{kd}=\sum_{h=1}^{k\land d}\binom{k}{h},
\end{align*}
which is the number of non-empty sets of at most $d$ servers. We need the following recurrence relation for $c_{kd}$.
\begin{lemma}\label{lem:recurrence}
	For $k\ge 0$ and $d\ge 1$,
\begin{align*}
c_{kd}&=k+\sum_{i=0}^{k-1}c_{i,d-1}.
\end{align*}
\end{lemma}
\begin{proof}
The statement follows by induction on $k$. For the induction step, we have
\begin{align*}
k+\sum_{i=0}^{k-1}c_{i,d-1}&=1+c_{k-1,d-1}+\left(k-1 +\sum_{i=0}^{k-2}c_{i,d-1}\right)\\
&= 1+c_{k-1,d-1} + c_{k-1,d}\\
&= \sum_{h=1}^{k\land d}\binom{k-1}{h-1} + \sum_{h=1}^{(k-1)\land d}\binom{k-1}{h}\\
&= \sum_{h=1}^{k\land d}\binom{k}{h} = c_{kd}.\qedhere
\end{align*}
\end{proof}

For a given point in time, fix a naming of the servers by the numbers $0,\dots,k-1$ such that their heights $h_0\le\dots \le h_{k-1}$ are non-decreasing. If we consider a (small) step, we choose this numbering such that $h_0\le\dots \le h_{k-1}$ holds both {\em before} and {\em after} the step. Since it is not possible that a server is strictly higher than another server before the step and then strictly lower afterwards, such a numbering exists. Let $U, B\subseteq\{0,\dots,k-1\}$ be the sets of servers that move upwards (i.e., towards the root) and downwards (away from the root), respectively. Note that by our earlier observation,
$B$ is either a singleton $\{j\}$ (one server moves downwards) or the empty set (no server moves downwards). We sometimes write a sum over elements of $B$, so such a sum is either $0$ or a single term.

Let $\alpha_\ell$ denote the weighted height of the node layer at combinatorial height $\ell$ in the HST (i.e., the distance of these vertices from the leaf layer). Thus, $0=\alpha_0<\alpha_1< \dots< \alpha_d$. We use the following potential function at time $t$:
\[\Psi_t=\sum_{i=0}^{k-1} \sum_{\ell=0}^{d-1} c_{i\ell} \cdot  \max\left\{\alpha_{\ell},h_{it}\land \alpha_{\ell+1}\right\},\]
where $h_{0t}\le\dots \le h_{k-1,t}$ are the weighted heights of the servers.
We next show that $\Psi$ satisfies the requirements of Observation \ref{obs:phireq} with $c=c_{kd}$.

\paragraph{A relocation request, $t\in T_r$:}
Since all requests are at the leaves and thus the height of the moving server is 0, we have $\Psi_t=\Psi_{t-1}$.

\paragraph{A small step of a simple request, $t\in T_s$:}
Let $\ell_i$ be such that the edge traversed by server $i$ during the step is located between the node layers of combinatorial heights $\ell_{i}$ and $\ell_{i+1}$. Thus, the weighted height of $i$ lies in $[\alpha_{\ell_i},\alpha_{\ell_i+1}]$ during the step. Then
\begin{align}
\Delta \Psi & = \sum_{i\in U}c_{i\ell_i} - \sum_{j\in B} c_{j\ell_j}. \label{eq:dPsi}
\end{align}
There are two cases to be considered.
\paragraph{Case 1: $B=\emptyset$:}
\begin{align}
|U| + \Delta \Psi & \leq k+\sum_{i=0}^{k-1} c_{i\ell_i} \label{ineq11}\\
& \leq  k+\sum_{i=0}^{k-1} c_{i,d-1} = c_{kd}. \label{ineq12}
\end{align}
Inequality \eqref{ineq11} follows from \eqref{eq:dPsi} and since $|U|\leq k$. Inequality \eqref{ineq12} follows from the definition of $c_{kd}$. The final equation is due to Lemma~\ref{lem:recurrence}.

\paragraph{Case 2: $B=\{j\}$:}

In this case $U\subseteq\{0,\dots, j-1\}$ and $\ell_i\le \ell_j-1$ for each $i\in U$.
Therefore,
\begin{align}
|U| + \Delta \Psi & \le j - c_{j\ell_j} + \sum_{i\in U}c_{i\ell_i} \label{ineq21}\\
& \leq j- c_{j\ell_j} +\sum_{i=0}^{j-1} c_{i,\ell_j-1} =0.\label{ineq22}
\end{align}
Inequality \eqref{ineq21} follows from \eqref{eq:dPsi} and since $U\subseteq\{0,\dots, j-1\}$. Inequality \eqref{ineq22} follows from the definition of $c_{kd}$. The equation is due to Lemma~\ref{lem:recurrence}.

%% file: dc-weighted-trees.tex
\section{\texorpdfstring{The $k$-taxi Problem on Weighted Trees}{The k-taxi Problem on Weighted Trees}}\label{sec:generaltree}
In this section we analyze \DC on general (weighted) trees, proving part \ref{it:ThmUbTree} of Theorem \ref{thm:hst}.
Specifically, we prove the following theorem.
\begin{theorem}\label{thm:weightedacc}
The competitive ratio of \DC on weighted trees of depth $d$ is at most
\begin{align*}
&4d-1&&\text{if }k=2\\
&\frac{2k(k-1)^d-3k+2}{k-2}=O(k^d)&&\text{if }k\ge 3.
\end{align*}
\end{theorem}

\noindent There are two reasons why our analysis for HSTs fails on general weighted trees:
\begin{enumerate}
	\item The costs of movement towards and away from the root no longer need to be within a constant of each other. E.g., if relocation repeatedly brings servers closer to the root, then most cost would be incurred while moving away from the root.
	\item The potential is no longer constant under relocation requests, because servers can be relocated to and from internal vertices, affecting their height.
\end{enumerate}

To address the first issue, we use the LP formulations that measure movement cost both towards and away from the root. As discussed in Section~\ref{sec:LP}, the only change in the dual is to replace $A_{ut}-A_{p(u)t} \in [0,1]$ by  $A_{ut}-A_{p(u)t} \in [-1,1]$. To address the second issue, we will eliminate the potential function from our proof. Instead, we will construct a dual solution that bounds the cost of \DC in each step, but it may violate the constraints $A_{ut}-A_{p(u)t}\in[-1,1]$. However, it will still satisfy $A_{ut}-A_{p(u)t}\in[-c,c]$ for some $c$. Thus, dividing all dual variables by $c$ yields a feasible dual solution, and $c$ is our competitive ratio.

\subsection{Proof of Theorem \ref{thm:weightedacc}}

We now turn to a more detailed description of the analysis. The proof uses the same notation and the observations as in Section~\ref{sec:hst}.
Let $d$ be the depth of the tree. For $i=1,\dots,d$, let
\begin{align*}
m_i&:= \begin{cases}
-2(d-i)-1&\text{if }k=2\\
\frac{-2(k-1)^{d-i+1}+k}{k-2}&\text{if }k\ge 3
\end{cases}\\
M_i&:= \begin{cases}
2(d+i)-1&\text{if }k=2\\
\frac{2k(k-1)^d-2(k-1)^{d-i+1}-k}{k-2}&\text{if }k\ge 3.
\end{cases}
\end{align*}

These quantities have been chosen to satisfy the following lemma:
\begin{lemma}\label{lem:weightedMm}
The values $m_i$ and $M_i$ satisfy the following:
	\begin{enumerate}[(a)]
		\item $m_1<m_2<\dots<m_d=-1<M_1<M_2<\dots<M_d$\label{it:mIncr}
		\item $M_i-m_i$ is a constant independent of $i$\label{it:M-mConst}
		\item $M_1+(j-1)m_1\ge j$ for all $j=1,\dots,k$\label{it:Mm}
		\item $(j-1) m_{i+1}-m_{i}\ge j$ for all $j=1,\dots,k$ and $i=1,\dots,d-1$\label{it:mm}
	\end{enumerate}
\end{lemma}
\begin{proof}
	Properties \ref{it:mIncr} and \ref{it:M-mConst} are straightforward to check. For properties \ref{it:Mm} and \ref{it:mm}, since $m_1<m_{i+1}<0$ it suffices to show these for the case $j=k$. For $k=2$, they are easily verified. For $k\ge 3$, we have
	\begin{align*}
	M_1+(k-1)m_1&=\frac{2k(k-1)^d-2(k-1)^{d}-k}{k-2} + (k-1)\frac{-2(k-1)^{d}+k}{k-2}\\
	&= \frac{k^2-2k}{k-2}=k\\
	(k-1)m_{i+1}-m_{i}&=(k-1)\frac{-2(k-1)^{d-i}+k}{k-2} - \frac{-2(k-1)^{d-i+1}+k}{k-2}\\
	&= \frac{k^2-2k}{k-2}=k\qedhere
	\end{align*}
\end{proof}

Recall that we break the long edges of the original tree into short edges of unit length. Requests arrive only in the subset of $V$ that are endpoints of long edges. For a node $u\in V$ (which might lie in the middle of a long edge), define its depth $d_u$ to be be the minimal number of long edges of a path of long edges that starts at the root $\rt$ and includes $u$. We will construct a dual solution that satisfies constraint \eqref{eq:AltGrow}, but instead of \eqref{eq:slope} it will only satisfy $A_{ut}-A_{p(u)t}\in[m_{d_u},M_{d_u}]$. We use again the terminology and notation from before. To satisfy these (relaxed) constraints, we impose the following two rules on the extended dual solution that we will be constructing:
\begin{enumerate}[(i)]
	\item $\Delta A_{u} \ge 0$ for all $u\in V$: Altitudes can only decrease in reverse-time.
	\item $A_{u}-A_{p(u)} \in[m_{d_u},M_{d_u}]$ for all $u\ne \rt$ at all times: Adjacent altitudes are not too different.
\end{enumerate}
 We will show for all $t\in T_s\cup T_r$ that
\begin{align}
\cost_t\le D_t,\label{eq:weightedMain}
\end{align}
where $\cost_t$ is the movement cost by \DC to serve the $t$th request.

\paragraph{The competitive ratio:} A feasible dual solution can then be obtained by dividing all altitudes by
\begin{align*}
c:=\max\{M_i,|m_i|\colon i=1,\dots,d\} = M_d=\begin{cases}
4d-1&\text{if }k=2\\
\frac{2k(k-1)^d-3k+2}{k-2}=O(k^d)&\text{if }k\ge 3.
\end{cases},
\end{align*}
As this also divides the dual objective value by $c$, it implies that \DC is $c$-competitive, proving Theorem \ref{thm:weightedacc}.

\paragraph{Constructing a dual solution satisfying \eqref{eq:weightedMain}:}
As before, we proceed by induction backwards in time, and divide into several cases. We start with some arbitrary fixed altitude $A$ for each vertex at the time after the final request.

\paragraph{A relocation request, $t\in T_r$:}
We keep all altitudes unchanged. Observe that \eqref{eq:weightedMain} is satisfied with both sides equal to $0$.

\paragraph{A simple request, $t\in T_s$:}
We break the movement again into small steps where \DC's servers move by distance $1$ and employ the earlier notation. For a given step, the cost of \DC is $|U|+|B|$. Thus, we obtain \eqref{eq:weightedMain} if we can show for the step that
\begin{align*}
|U| + |B|  &\le \Delta D.
\end{align*}

\paragraph{Case 1: $B=\emptyset$:}
Let $\delta:=\min_{i\in U} M_{d_{v_i'}}-A_{v_i'}+A_{p(v_i')}$ and let the minimum be achieved for $i^*\in U$. By rule \ref{it:r2} of the induction hypothesis, we have $\delta\ge 0$. Therefore, the following reverse-time update of altitudes satisfies rules~\ref{it:r1} and \ref{it:r2}:
\begin{align*}
A_{u}' & := A_u - \delta &&\text{for } u \in V \setminus \cup_{i\in U}V_{v_i}\\
A_{u}' & := A_u  &&\text{for } u \notin V \setminus \cup_{i\in U}V_{v_i}
\end{align*}
We have
\begin{align}
\Delta D & = \Delta A_{s_t}-\sum_{i\in U} \Delta A_{i} \nonumber \\
&= \delta+\sum_{i\in U}\left(A_{v_i'} - A_{p(v_i')}\right) \label{eqw0}\\
& \ge M_{d_{v_{i^*}'}}+\sum_{i\in U\setminus\{i^*\}}m_{d_{v_i}} \label{ineqw1}\\
&\ge M_1+(|U|-1)m_1 \label{ineqw2}\\
&\ge |U|=|U|+|B|. \label{ineqw3}
\end{align}
Equation \eqref{eqw0} follows since the request location $s_t$ lies in the component where altitudes are changed by $\delta$ and the updated server positions $v_i'$ are all outside of it.
Inequality \eqref{ineqw1} follows from rule~\ref{it:r2} of the induction hypothesis.
Inequalities \eqref{ineqw2} and \eqref{ineqw3} follow by Lemma~\ref{lem:weightedMm}.

\paragraph{Case 2: $B=\{j\}$:} Let $\delta\ge 0$ be the minimum of the set $\{A_{v_j}-A_{p(v_j)}-m_{d_{v_j}}\}\cup\{M_{d_{v_i'}}-A_{v_i'}+A_{p(v_i')}\mid i\in U\}$. We modify the altitudes as follows.
\begin{align*}
A_{u}' & := A_u - \delta &&\text{for } u \in V_{v_j} \setminus \cup_{i\in U}V_{v_i}\\
A_{u}' & := A_u  &&\text{for } u \notin V_{v_j} \setminus \cup_{i\in U}V_{v_i}
\end{align*}
Again, rules \ref{it:r1} and \ref{it:r2} are obeyed.
We have
\begin{align}\allowdisplaybreaks
\Delta D & = \Delta A_{s_t}-\Delta A_{j} - \sum_{i\in U} \Delta A_{i}  \nonumber\\
& = \delta + A_{p(v_j)}-A_{v_j} + \sum_{i\in U}\left(A_{v_i'} - A_{p(v_i')}\right) \nonumber\\
&= \begin{cases}
-m_{d_{v_j}}+ \sum_{i\in U}\left(A_{v_i'} - A_{p(v_i')}\right)\qquad&\text{if }\delta=A_{v_j}-A_{p(v_j)}-m_{d_{v_j}} \nonumber\\
M_{d_{v_{i^*}}}+A_{p(v_j)}-A_{v_j} + \sum_{i\in U\setminus \{i^*\}}\left(A_{v_i'} - A_{p(v_i')}\right)\qquad&\text{if $\delta=M_{d_{v_{i^*}'}}-A_{v_i'}+A_{p(v_{i^*}')}$ for $i^*\in U$} \nonumber
\end{cases} \nonumber\\
&\ge \begin{cases}
-m_{d_{v_j}}+ \sum_{i\in U}m_{d_{v_i'}}\nonumber\\
M_{d_{v_{i^*}}} - M_{d_{v_j}} + \sum_{i\in U\setminus \{i^*\}}m_{d_{v_i}} \nonumber
\end{cases} \nonumber\\
&\ge \begin{cases}
-m_{d_{v_j}}+ |U|m_{d_{v_j}+1} \\
M_{d_{v_j}+1} - M_{d_{v_j}} + (|U|-1)m_{d_{v_j}+1}
\end{cases}  \label{eqStar}\\
&=-m_{d_{v_j}}+ |U|m_{d_{v_j}+1}\label{eq2w1}\\
&\ge |U|+1=|U|+|B|. \label{ieq2w1}
\end{align}
Inequality \eqref{eqStar} follows by Lemma~\ref{lem:weightedMm}\ref{it:mIncr} and since $d_{v_i'}\ge d_{v_j}+1$ for all $i\in U$.
Equation \eqref{eq2w1} follows by Lemma~\ref{lem:weightedMm}\ref{it:M-mConst}. Finally, inequality \eqref{ieq2w1} follows by Lemma~\ref{lem:weightedMm}\ref{it:mm} if $d_{v_j}<d$; if $d_{v_j}=d$, then $|U|=0$ and the inequality holds due to Lemma~\ref{lem:weightedMm}\ref{it:mIncr}, which states that $m_d=-1$. This concludes the proof of inequality \eqref{eq:weightedMain}, and thereby Theorem~\ref{thm:weightedacc}. 

%% file: lowerbound.tex
\section{Lower Bounds}\label{sec:Lb}

In this section we show lower bounds on the competitive ratio of \DC, proving Theorem \ref{thm:lowerbound}.

\subsection{\texorpdfstring{Lower Bound for Depth $d$ Trees}{Lower Bound for Depth d Trees}}
We will prove the following theorem, implying part \ref{it:ThmLbTree} of Theorem \ref{thm:lowerbound}.
\begin{theorem}\label{thm:wTreeLbPrecise}
The competitive ratio of \DC for the $k$-taxi problem on unweighted tree metrics of depth $d$ is at least
\begin{align*}
4\sum_{h=1}^{d-1}\binom{k+h-2}{h}+2\binom{k+d-2}{d}+1.
\end{align*}
\end{theorem}
Notice that for $d=1$, the lower bound is equal to $2k-1$, and for $k=2$ it is equal to $4d-1$, and both of these cases match exactly our upper bound from Theorem~\ref{thm:weightedacc}. For constant $d$, the lower bound is at least $\Omega(k^d)$ as $k\to\infty$, matching our upper bound for weighted trees up to a constant depending on $d$.

We prove the lower bound on the $k$-ary tree of depth $d$, where each edge has length $1$. We call a pair of online and offline configurations a \emph{situation}. Consider a situation with the following properties:
\begin{itemize}
	\item The location of $j$ online servers matches that of $j$ offline servers.
	\item Of the remaining online servers, at least one is located at a vertex $x$, and none is in the subtree below $x$.
	\item Of the remaining offline servers, at least one is located at a child $y$ of $x$.
\end{itemize}
We call such a situation a \emph{$j$-match around $(x,y)$}.

\begin{lemma}\label{lem:wTreeLbjMatch}
Suppose the current situation is a $j$-match around $(x,y)$ and let $h$ be the height of $y$. There exists a request sequence on which \DC suffers cost $2\binom{j+h}{h}-1$, the offline algorithm suffers cost $0$ and the resulting situation differs from the original one only in that there is one online server less at $x$ and instead there is now an online server at $y$.
\end{lemma}
\begin{proof}
	We proceed by induction on $h$.
	
	If $h=0$, we issue a single simple request at $y$. Upon this request, \DC moves a single server from $x$ to $y$ for cost $1=2\binom{j}{0}-1$. The request is free for the offline algorithm and the difference between the resulting and original situation is as desired.
	
	\begin{figure}
		\begin{center} 
			\includegraphics{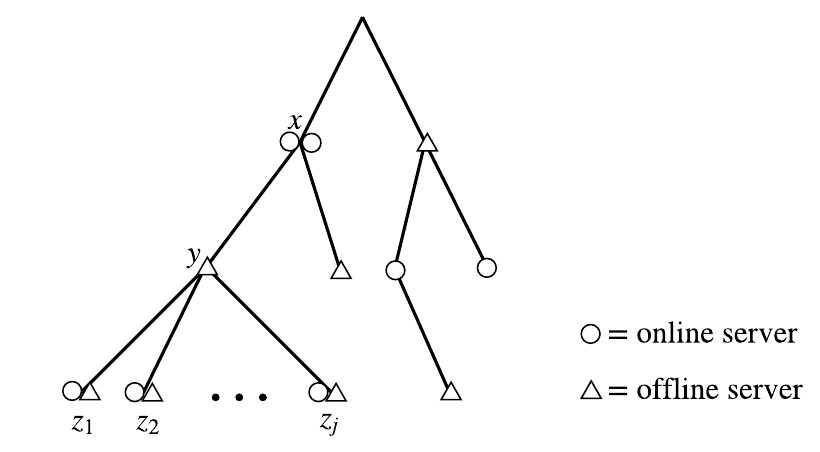}
			\caption{A $j$-match around $(x,y)$ (not all edges depicted).}
			\label{fig:jmatch}
		\end{center}
	\end{figure}
	
	If $h\ge 1$, we can assume (by issuing some relocation requests at the beginning and end of the request sequence) that the $j$ pairs of matching online and offline servers are located at $j$ children $z_1,\dots,z_j$ of $y$ (see Figure~\ref{fig:jmatch}). We first issue a request at $y$, which moves the $j$ online servers from $z_1,\dots,z_j$ up to $y$ and one online server from $x$ down to $y$, overall incurring cost $j+1$. Notice that the new situation is a $1$-match around $(y,z_i)$ for each $i=1,\dots,j$. We apply the induction hypothesis $j$ times; we will maintain the invariant that before the $\ell$th application, we are in an $\ell$-match around $(y,z_i)$ for each $i=\ell,\dots,j$. By applying the $\ell$th application to the $\ell$-match around $(y,z_\ell)$, the invariant is indeed maintained. After the last application of the induction hypothesis, we have a situation that differs from the original one in that an online server got removed from $x$ and added to $y$, as desired. The offline cost of the sequence is $0$ and the cost of \DC is
	\begin{align*}
	j+1+\sum_{\ell=1}^j\left[2\binom{\ell+h-1}{h-1}-1\right]&= 1+2\sum_{\ell=1}^j\binom{\ell+h-1}{h-1}=2\binom{j+h}{h}-1,
	\end{align*}
	where the last equation follows from the identity
	\begin{align*}
	\sum_{\ell=0}^j\binom{\ell+h-1}{h-1}=\binom{j+h}{h}.&\qedhere
	\end{align*}
\end{proof}

We call a situation \emph{$(h,\uparrow)$-situation} (resp. \emph{$(h,\downarrow)$-situation}) if the location of $k-1$ online servers matches that of $k-1$ offline servers, the last online server is at a node $x$ at height $h$ and the last offline server is at the parent (resp. a child) of $x$. We say a \emph{situation transforms to another situation at cost $c$} if there exists a request sequence that leads from the first to the second situation, incurring cost $c$ for \DC and cost $0$ for the offline algorithm.

\begin{lemma}\label{lem:wTreeLbUpDown}
	Let $b_h=2\binom{k+h-1}{h+1}$.
	\begin{enumerate}[(a)]
		\item For $h=0,\dots,d-2$, any $(h,\uparrow)$-situation transforms to a $(h+1,\uparrow)$-situation at cost $b_h$.
		\item Any $(d-1,\uparrow)$-situation transforms to to a $(d,\downarrow)$-situation at cost $b_{d-1}$.
		\item For $h=2,\dots,d$, any $(h,\downarrow)$-situation transforms to a $(h-1,\downarrow)$-situation at cost $b_{h-2}$.\label{it:wTreeCaseDown}
	\end{enumerate}
\end{lemma}
\begin{proof}
	\begin{enumerate}[(a)]
		\item Denote by $x$ the vertex where the unmatched online server is located and by $y$ the parent of $x$ where the unmatched offline server is located. We first issue some relocation requests, so that one of the matching server pairs is at the parent $z$ of $y$ and the other $k-2$ matching server pairs are at distinct siblings $x_1,\dots,x_{k-2}$ of $x$ (see Figure~\ref{fig:hsituation}, left) Now request $y$. This results in all online servers moving to $y$ for cost $k$. The offline servers are still at $z,y,x_1,\dots,x_{k-2}$. Notice that for each $i=1,\dots,k-2$, the current situation is a $1$-match around $(y,x_i)$. We will apply Lemma~\ref{lem:wTreeLbjMatch} $k-2$ times: Before the $\ell$th application, the current situation is an $\ell$-match around $(y,x_i)$ for all $i=\ell,\dots,k-2$. We can maintain this invariant by applying Lemma~\ref{lem:wTreeLbjMatch} to the $\ell$-match around $(y,x_\ell)$. After all these applications of Lemma~\ref{lem:wTreeLbjMatch}, there are two online servers at $y$ (one of them matching the offline server at $y$) and the others are matching the offline servers at $x_1,\dots,x_{k-2}$. Since $y$ is at height $h+1$, and the last offline server at the parent $z$ of $y$, we are now in a $(h+1,\uparrow)$-situation. The total online cost of the transformation is
		\begin{align*}
		k+\sum_{\ell=1}^{k-2}\left[2\binom{\ell+h}{h}-1\right] = 2\sum_{\ell=0}^{k-2}\binom{\ell+h}{h}=2\binom{k+h-1}{h+1}=b_h.
		\end{align*}
		\label{it:wTreeCaseUp}
		\item The proof is identical to case \ref{it:wTreeCaseUp} except that vertex $z$ is now a child from $y$ distinct from $x,x_1,\dots,x_{k-2}$. (Notice that $y$ is the root, since $x$ is at height $d-1$.)
		\item In this case, the vertex $y$ where the unmatched offline server is located is a child of the vertex $x$. We first issue some relocation requests so that the $k-1$ matched server pairs are at children $z_1,\dots,z_{k-1}$ of $y$ (see Figure~\ref{fig:hsituation}, right). Then we request $y$, forcing all online servers to move to $y$ at cost $k$. The new situation is a $1$-match around $(y,z_i)$ for each $i=1,\dots,k-1$. We now apply Lemma~\ref{lem:wTreeLbjMatch} $k-2$ times, similarly to before, and after all these applications we reach a $k-1$-match around $(y,z_{k-1})$. Since $y$ is at height $h-1$, this is also a $(h-1,\downarrow)$-situation. The total cost is obtained by the same calculation as in case~\ref{it:wTreeCaseUp}, with $h$ replaced by $h-2$ since the vertices $z_i$ are at height $h-2$.\qedhere
	\end{enumerate}
\end{proof}
\begin{figure}
\begin{center} 
	\includegraphics{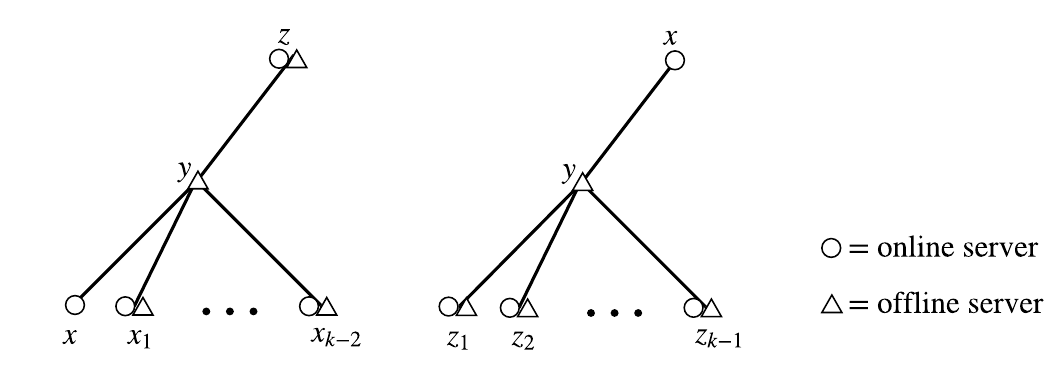}
	\caption{A $(h,\uparrow)$-situation (left) and $(h,\downarrow)$-situation (right), if $x$ is at height $h$.}
	\label{fig:hsituation}
\end{center}
\end{figure}

\begin{proof}[Proof of Theorem~\ref{thm:wTreeLbPrecise}]
	From a situation where \DC and the offline algorithm are in the same configuration, the offline algorithm can pay cost $1$ to reach a $(0,\uparrow)$-situation. Now we successively apply all cases of Lemma~\ref{lem:wTreeLbUpDown}, so that we eventually reach a $(h-1,\downarrow)$-situation. After one more request to the unmatched offline server, \DC pays an additional cost $1$, and the two algorithm are again in the same configuration. While the total cost of the offline algorithm was only $1$, the cost of \DC, and therefore a lower bound on its competitive ratio, is
	\begin{align*}
	2\sum_{h=0}^{d-2}b_h+b_{d-1}+1=4\sum_{h=1}^{d-1}\binom{k+h-2}{h}+2\binom{k+d-2}{d}+1.&\qedhere
	\end{align*}
\end{proof}

\subsection{Lower Bound for HSTs}

We now prove that our analysis on HSTs is exactly tight for any depth by giving a matching lower bound of $c_{kd}$ on the competitive ratio of \DC, which yields part \ref{it:ThmLbHst} of Theorem \ref{thm:lowerbound}. We remark that for $d=1$, the lower bound $c_{k1}=k$ already follows from the known lower bound on the $k$-server problem, and for $d\ge k$, the lower bound $c_{kd}=2^k-1$ follows from \cite{CK19}, where it was shown that even randomized algorithms against adaptive adversaries cannot achieve a better competitive ratio on HSTs of depth $d\ge k$.

For $\alpha\in\mathbb N$, let $T_{\alpha d}$ be an HST of depth $d$ with edge lengths $\alpha^{d-1},\alpha^{d-2},\dots,\alpha^0$ along each root-to-leaf path and where each internal vertex has sufficiently many (i,e., at least $k+1$) children. Let $W_{\alpha d}:=\sum_{h=0}^{d-1}\alpha^h=\frac{\alpha^{d}-1}{\alpha-1}$ be the distance from the root to any leaf. The lower bound in Theorem~\ref{thm:lowerbound} for HSTs follows from the following lemma by letting $\alpha\to\infty$.

\begin{lemma}
	For all $k\in\mathbb N_0$, $d\in\mathbb N_0$, $\alpha\in\mathbb N$, any initial configuration of $k$ servers in $T_{\alpha d}$ and any leaf $\ell$ of $T_{\alpha d}$, there exists a request sequence with the following properties when being served:
	\begin{enumerate}[(a)]
		\item The upwards movement cost of \DC is at least $(\alpha-1)^{d-1}c_{kd}$.\label{it:LbHstCost}
		\item The upwards movement cost of an optimal offline algorithm is at most $W_{\alpha d}$.\label{it:LbHstOpt}
		\item The cost of an algorithm with an additional $(k+1)$st server at $\ell$ is $0$.\label{it:LbHstFree}
		\item If \DC had an additional $(k+1)$st server sufficiently far away from the root on an extra edge incident to the root, this server would move distance $W_{\alpha d}$ towards the root.\label{it:LbHstPull}
	\end{enumerate}
\end{lemma}
\begin{proof}
	We proof the lemma by induction on $d$. For $d=0$, the empty request sequence trivially yields the result since $c_{k0}=W_{\alpha0}=0$.
	
	Consider now the case $d\ge 1$. Denote by $S_0,S_1,\dots,S_k$ different depth-$(d-1)$-subtrees, where $S_0$ is the one containing $\ell$. Notice that each $S_i$ is a copy of $T_{\alpha,d-1}$. We first issue relocation requests to ensure that for each $i=1,\dots,k$, there is a server at some leaf $\ell_i$ in $S_i$. We then issue a request at $\ell$. To serve this request, the offline algorithm moves its server from $\ell_k$ to $\ell$ for upwards movement cost $W_{\alpha d}$, and it will suffer no additional cost for the remainder of the request sequence. In particular, this will ensure that an algorithm with an additional server at $\ell$ would suffer $0$ cost, as required. \DC moves its $k$ servers to the root for cost $kW_{\alpha d}$ and then moves one server down to $\ell$ to serve the request. We will construct the remainder of the request sequence so that there are times $t_1<\dots<t_k$ such that at time $t_i$, the following holds:
	\begin{itemize}
		\item There are $i$ online servers whose positions match those of $i$ offline servers.
		\item The remaining $k-i$ online servers are at the root.
		\item The remaining $k-i$ offline servers are at $\ell_{i},\dots,\ell_{k-1}$.
	\end{itemize}
	Notice that initially, these properties are satisfied for $i=1$. For $i<k$, the requests between times $t_i$ and $t_{i+1}$ are as follows: First we issue $i$ relocation requests so that the $i$ matching server pairs are in $S_i$. Now we issue the request sequence from the induction hypothesis applied to the subtree $S_i$, with $i$ instead of $k$ servers and with $\ell$ replaced by the vertex $\ell_i$ where the extra offline server is located. By property \ref{it:LbHstFree} of the induction hypothesis, this incurs no additional cost for the offline algorithm. By property \ref{it:LbHstPull} of the induction hypothesis, it will cause one of the online servers from the root of $T_{\alpha d}$ to move towards the root of $S_i$ by distance $W_{\alpha,d-1}$. Thus, we can run the sequence from the induction hypothesis $\alpha-1$ times and the server will move distance $\alpha^d-1$ from the root of $T_{\alpha d}$ towards the root of $S_i$ and hence it still has not reached $S_i$. In the end, we request the $i+1$ offline server locations in $S_i$ repeatedly until \DC has a server at all these locations. Notice that the properties for time $t_{i+1}$ are now satisfied. This completes the description of the request sequence.
	
	The upwards movement cost of \DC during the $\alpha-1$ invocations of the induction hypothesis between times $t_i$ and $t_{i+1}$ is at least $(\alpha-1)^{d-1}c_{i,d-1}$ due to property \ref{it:LbHstCost} of the induction hypothesis. Thus, the upwards movement cost during the entire request sequence is at least
	\begin{align*}
		kW_{\alpha d}+(\alpha-1)^{d-1}\sum_{i=1}^{k-1}c_{i,d-1}\ge (\alpha-1)^{d-1}c_{kd},
	\end{align*}
	where we have used $W_{\alpha d}\ge (\alpha-1)^{d-1}$ and the recurrence from Lemma~\ref{lem:recurrence}. This proves property \ref{it:LbHstCost}. As announced, the offline algorithm does not incur any additional cost beyond moving from $\ell_{k}$ to $\ell$ in the beginning, and therefore \ref{it:LbHstOpt} and \ref{it:LbHstFree} are also satisfied. Property \ref{it:LbHstPull} holds because such an additional online server ``above the root'' would be pulled down only during the initial stage where the $k$ online servers move upwards by distance $W_{\alpha d}$ each.
\end{proof} 

%% file: limitations.tex
\section{Limitation of Previous Techniques}\label{sec:limit}

\subsection{Forward-Time Primal Dual}\label{sec:noForward}
We show here that a forward-time construction of a dual solution would not have allowed us to obtain our results, because it would be unable to approximate the optimal value in general.

\begin{proof}[Proof of Theorem~\ref{thm:forward}]
	Consider the tree containing only two leaves that are at distance $1$ from the root. Consider an arbitrary deterministic online algorithm for the dual LP. (If the algorithm is randomized, then the same proof works by replacing altitudes with expected altitudes, so this is without loss of generality.) We will construct a request sequence for which the dual algorithm fails to achieve any positive objective value even though the optimal dual objective value (which is equal to the optimal objective value of the $k$-taxi problem) tends to infinity. Fix $k=1$, and denote the single \DC server by $i$. Note that this server is always located at the destination of the last request. The following proof easily extends to greater values of $k$ by relocating \emph{all} servers to the same leaf before issuing a simple request at the other leaf.\footnote{If one does not want to relocate several servers to the same vertex, one can also expand the two leaves to subtrees of small diameter and relocate the servers to the same subtree.}
	
	Consider the first time $t$ for which no request has been issued yet. Denote the two leaves by $u$ and $w$ such that $A_{u,t-1}\le A_{w,t-1}$.
	
	\paragraph{If the server is at $w$ just before time $t$:} We first issue a relocation request from $s_t=w$ to $d_t=u$ and then a simple request at $s_{t+1}=w$. At time $t$, the objective value changes by
	\begin{align*}
	D_t=-\Delta_t A_{u}=A_{u,t-1}-A_{ut},
	\end{align*}
	and at time $t+1$, it changes by
	\begin{align*}
	D_{t+1}=\Delta_{t+1} A_w - \Delta_{t+1} A_i = A_{w,t+1} - A_{wt} - A_{w,t+1} + A_{ut} = A_{ut}-A_{wt}.
	\end{align*}
	So overall, the objective value changes by
	\begin{align*}
	D_t+D_{t+1} = A_{u,t-1}-A_{wt}\le A_{u,t-1}-A_{w,t-1}\le 0,
	\end{align*}
	where the first inequality is due to constraint \eqref{eq:AltGrow} and the other inequality follows by definition of $u$ and $w$.
	
	\paragraph{If the server is at $u$ just before time $t$:}
	Then we issue a simple request at $s_t=w$. The objective value changes by
	\begin{align*}
	D_{t}=\Delta_{t} A_w - \Delta_{t} A_i = A_{wt} - A_{w,t-1} - A_{wt} + A_{u,t-1} = A_{u,t-1}-A_{w,t-1}\le 0.
	\end{align*}
	
	Thus, the dual objective value never increases, but any $1$-taxi algorithm has to pay a constant movement cost for each simple request.
\end{proof}

Notice that the request at time $t$ is chosen based on the altitudes at time $t-1$. When altitudes are constructed backwards in time, such an adversarial request sequence would not be well-defined.

\subsection{Matching Potential}\label{sec:purePotential}
Coester and Koutsoupias~\cite{CK19} gave a \emph{randomized} algorithm for the $k$-taxi problem on HSTs which also achieves competitive ratio $2^k-1$. The proof of competitiveness of this algorithm is given by a potential function argument, where the potential function is $2^k-1$ times the value of a minimum matching between the online and offline servers. In a sense, the algorithm can be viewed as the randomized analogue of \DC. Therefore, it is unsurprising that the same potential can be used to prove the same competitive ratio for \DC when $k=2$, as stated in~\cite{CK19}. However, they also mention that this potential fails for $k=3$. Indeed, consider the depth-1-HST with four leaves $a, b, c, d$ that are at distance $1$ from the root $\rt$. Consider the configuration with online servers at $a, b, c$ and offline servers at $b,c,d$. The minimum matching has value $2$. After a simple request at leaf $d$, \DC has one server at $d$ and the other two servers reside at the root $\rt$, while the offline configuration is unchanged. Observe that the new minimum matching still has value $2$. Thus, the online algorithm incurred cost for this request, but the potential and offline cost remain unchanged. Therefore, this potential cannot be used to prove competitiveness of $\DC$ on HSTs, even for $k=3$. We were unable to find a pure potential proof that proves competitiveness of \DC beyond the case $k=2$.

\section{The Transformed Dual's Dual}
Instead of the primal LP defined in Section~\ref{sec:LP}, we could have also defined the following different primal LP. It is less intuitive than the primal of Section~\ref{sec:LP}, but it directly yields the altitude LP as its dual.

For $t\in T_s$ and $u\ne \rt$, we use variables $y_{ut}$ and $z_{ut}$ for the numbers of servers leaving and entering subtree $V_u$ at time $t$, as before. For $u=\rt$, we view $y_{\rt t}$ and $z_{\rt t}$ as the constant $0$. We also use variables $x_{ut}$, but with a different meaning than before: Now, $x_{ut}$ denotes the number of servers \emph{at} vertex $u$ at time $t$ that are \emph{not} currently serving a simple request. Thus, $x_{ut}+\1_{\{u=s_t\text{ and }t\in T_s\}}$ is the total number of servers at vertex $u$ at time $t$. We no longer use variables for the number of servers within a subtree. The following LP models the $k$-taxi problem on trees.

\begin{alignat*}{3}
&&\text{min }~~~   & \sum_{t\in T_s} \sum_{u\ne \rt} (y_{ut}+z_{ut})\\
&\dl{(A_{u,t-1})}& & y_{ut}-z_{ut} - \sum_{v\prec u} (y_{vt}-z_{vt}) = x_{u,t-1}+\1_{\{u=s_{t-1}\text{ and }t-1\in T_s\}} - x_{ut}-\1_{\{u=s_t\}}  && \forall u\in V, t\in T_s \\
&\dl{(A_{u,t-1})}&& x_{ut} = x_{u,t-1} + \1_{\{u=d_t\}} &&  \forall u\in V, t\in T_r \\
&&& y_{ut}, z_{ut}\ge 0 && \forall u\ne \rt, t \in T_s\\
&&& x_{ut}\ge 0 && \forall u\in V, t \in T_s\cup T_r
\end{alignat*}
In the first constraint, the LHS and RHS are two different ways of writing the number of servers leaving vertex $u$ at time $t\in T_s$. In particular, on the RHS we subtract the new number of servers at $u$, namely $x_{ut}+\1_{\{u=s_t\}}=x_{ut}+\1_{\{u=s_t\text{ and }t\in T_s\}}$, from the old number $x_{u,t-1}+\1_{\{u=s_{t-1}\text{ and }t-1\in T_s\}}$. Since $x_{ut}\ge 0$, it is guaranteed that there is at least one server at the requested location of each simple request. The second constraint guarantees that for $t\in T_r$, the server that was previously serving the simple request at $s_{t-1}=d_{t-1}=s_t$ is now located at $d_t$ instead, and other servers remain at their old locations.

Before we construct the dual LP we will add the additional constraint
\begin{align*}
\dl{(A_{u,2T})}\qquad\qquad 0=x_{u,2T}+\1_{\{u=s_{t}\text{ and }t\in T_s\}}- \bar{x}_{u,2T} \qquad \forall u\ne\rt
\end{align*}
to the primal LP, where $\bar{x}_{u,2T}$ are \emph{constants} specifying the configuration of \DC at the last time step. As before, this affects the optimal value by only an additive constant. We again also view $x_{u0}=\bar{x}_{u0}$ as constants describing the initial configuration of the servers (with the altered meaning of the variables $x_{ut}$). We can write the second primal constraint slightly more complicated by adding $\1_{\{u=s_{t-1}\text{ and }t-1\in T_s\}} - \1_{\{u=s_{t}\}}$ to its RHS, which is $0$ because every relocation request at time $t\in T_r$ is preceded by a simple request at $s_t=s_{t-1}$. Then we obtain the following corresponding dual:
\begin{align*}
\lefteqn{\hspace*{-1cm}\max~~~ \sum_{t\in T_s}\left[A_{s_t t} - A_{s_t,t-1}\right] + \sum_{t\in T_r} \left[A_{d_t ,t-1} - A_{s_t, t-1}\right] + \sum_{u\in V}\left[\bar{x}_{u0}A_{u0} -\bar{x}_{u,2T}A_{u,2T} \right]}  \nonumber \\ \nonumber\\
& A_{ut}-A_{p(u)t} \in[-1, 1] &  \forall u\ne \rt, t+1 \in T_s \\
& \Delta_tA_{u}\ge 0\quad&  \forall u\in V, t\in T_s\cup T_r.
\end{align*}
Expanding the term $\sum_{u\in V}\left[\bar{x}_{u0}A_{u0} -\bar{x}_{u,2T}A_{u,2T} \right]$ to a telescoping sum, we exactly recover the dual objective function from before.